\documentclass[conference]{IEEEtran}
\IEEEoverridecommandlockouts
\usepackage{cite}
\usepackage{amsmath,amssymb,amsfonts}

\usepackage{algorithmic}
\usepackage{graphicx}
\usepackage{textcomp}
\usepackage{xcolor}
\def\BibTeX{{\rm B\kern-.05em{\sc i\kern-.025em b}\kern-.08em
    T\kern-.1667em\lower.7ex\hbox{E}\kern-.125emX}}
\usepackage{euscript}
\usepackage{bbm}

\usepackage{amsthm}
\usepackage{flushend}
  
\newtheorem{theorem}{Theorem}

\newtheorem{corollary}{Corollary}

\makeatletter
\def\ps@IEEEtitlepagestyle{%
  \def\@oddfoot{\mycopyrightnotice}%
  \def\@evenfoot{}%
  }

\def\mycopyrightnotice{%
  {\footnotesize 978-1-7281-0858-2/19/\$31.00 © 2019 IEEE \hfill}%
  \gdef\mycopyrightnotice{}%
  }
 
\begin{document}

\title{Data Replication for Reducing Computing Time in Distributed Systems with Stragglers
}

\author{\IEEEauthorblockN{Amir Behrouzi-Far and Emina Soljanin}
\IEEEauthorblockA{\textit{Department of Electrical and Computer Engineering, Rutgers University} \\
Piscataway, New Jersey 08854, USA \\
\{amir.behrouzifar,emina.soljanin\}@rutgers.edu}
}

\maketitle

\begin{abstract}
In distributed computing systems with stragglers, various forms of redundancy can improve the average delay performance. We study the optimal replication of data in systems where the job execution time is a stochastically decreasing and convex random variable. We show that in such systems, the optimum assignment policy is the balanced replication of disjoint batches of data. Furthermore, for Exponential and Shifted-Exponential service times, we derive the optimum redundancy levels for minimizing both expected value and the variance of the job completion time. Our analysis shows that, the optimum redundancy level may not be the same for the two metrics, thus there is a trade-off between reducing the expected value of the completion time and reducing its variance.
\end{abstract}

\section{Introduction}

Distributed computing has received a great attention for several reasons, including big data processing \cite{dean2004mapreduce} where parallelization is essential. Implementing computing algorithms in distributed systems introduces new challenges that have to be addressed in order to benefit from the paralleziation. In particular, since the failure rate and/or slowdown of a system increase with the number of computing nodes, robustness is an essential part of any reliable distributed computing/storage algorithm \cite{dean2013tail}.
For achieving robustness, redundancy is proposed in the literature \cite{brun2011smart}. Redundancy enables a task master to generate the overall result form the computations of only a subset of all computing nodes, instead of all of them. Thus, slow workers, known as stragglers, can be mitigated \cite{behrouzi2018effect}.

In spite of a growing body of work on redundancy in distributed systems, we are still far from fully understanding its benefits and costs. In particular stragglers, that are the consequence of e.g., resource contention, network congestion, Input/Output (I/O) operations, could be exacerbated by a non-careful redundancy planning \cite{aktas2019straggler}. Although understanding the exact benefits and costs requires detailed knowledge about both the system and the computing job, many studies have been devoted to performance evaluation of these systems under reasonable modeling, see
\cite{aktas2019learning} and references therein.

In this work, we consider a distributed system with the master-worker architecture. A computing job is an executable file which has to operate over a possibly large data set. We assume that the executable can be concurrently run on different workers, each hosting a subset of the original data set. This model is well applicable to a wide range of problems, e.g model training in supervised machine learning \cite{tandon2016gradient}. We study the completion time in the aforementioned system, defined by the total waiting time for the master to generate the overall results from a subset of local compute results. We study two general methods for data replication across workers: overlapping batches and non-overlapping batches. With overlapping batches, the data subset at each worker either completely overlaps or does not overlap at all with the subsets at other users. With overlapping batches, on the other hand, the data set at any worker has partial overlaps with subsets at other workers. We first showed that a balanced distribution of non-overlapping batches results in minimum expected completion time, if the service time distribution of workers is stochastically decreasing and convex random variable. We then studied the effect of the degree of redundancy on the completion time, for two distributions of service time of data samples: Exponential and Shifted Exponential. Our analysis show that, with Exponential service time distribution,  maximum diversity, i.e. replicating all the data at every worker, minimizes the completion time. On the other hand, with Shifted-Exponential service time distribution, neither full diversity nor full parallelism, i.e., dividing the data evenly among workers, is optimal and the optimum point depends on the distribution parameters.

\section{System model and Problem Statement}
We study a distributed computing system with a task-master and $N$ worker nodes, as given in Fig.~\ref{fig:sysModel1}. 
\begin{figure}[htbp]
   \centering
   \includegraphics[width=\columnwidth, keepaspectratio]{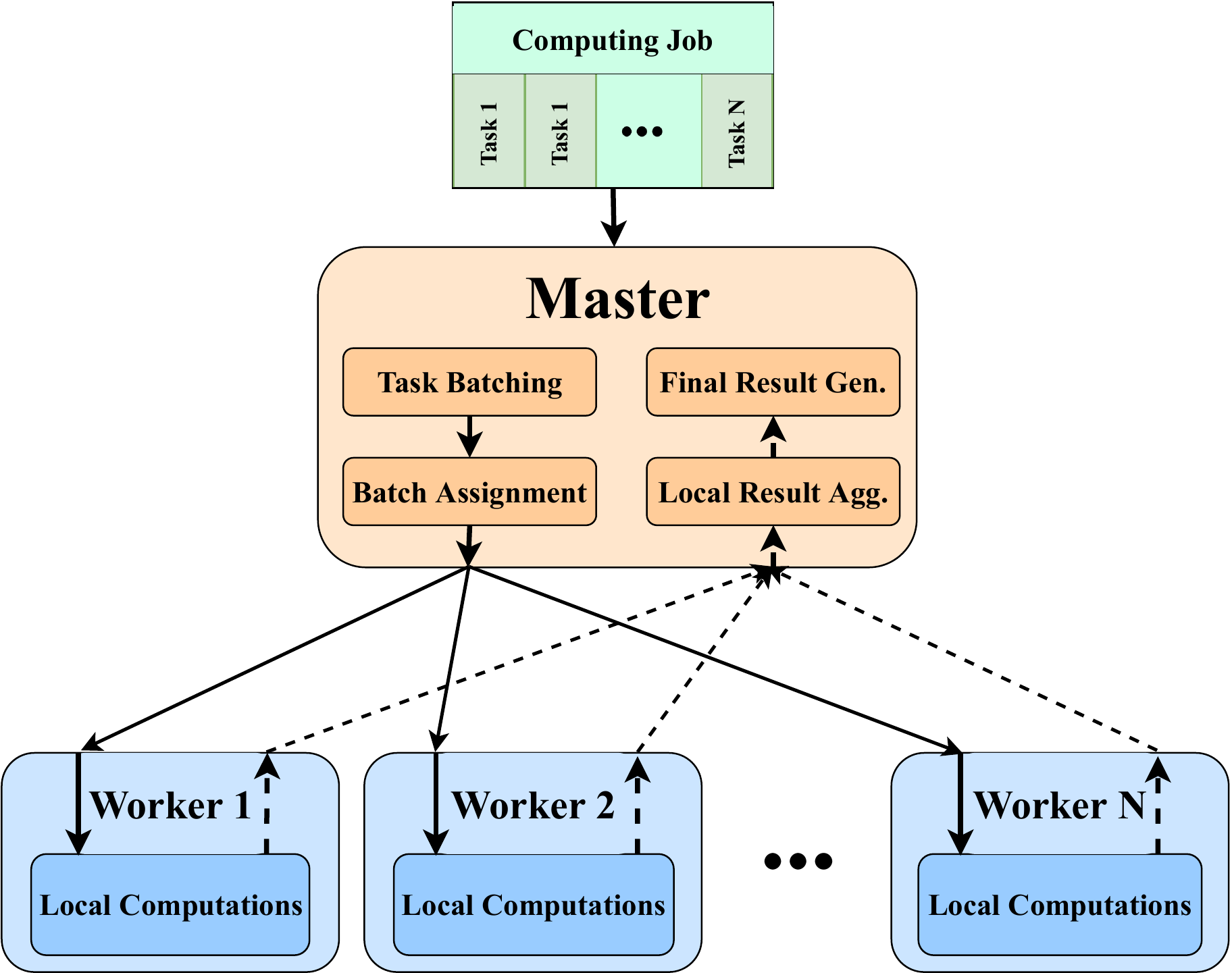}
   \caption{System model.}
   \label{fig:sysModel1}
\end{figure}
We refer to this system as System\ref{fig:sysModel1}. We assume that  computing jobs that are parallelizable to $N$ tasks. The master node hosts a batching unit, which puts data samples into $B$ batches, and a batch assignment unit, which allocates data batches to $N$ workers. There is also an aggregation unit in the master node, responsible for collecting local computations from worker nodes and sending them to result generation unit, which computes the overall result.

Data batches are redundantly assigned to multiple workers. Each worker sends its compute results to the master, only after finishing the execution of its compute job over the assigned data batch. Note that, sending the compute results at-once reduces the number of communications rounds. After receiving the compute results from a large enough number of workers, the master generates the overall result. This is a well applicable model to several applications, e.g., matrix multiplication \cite{choi1998new}, gradient based optimizers \cite{tandon2016gradient,boyd2011distributed} and model training in machine learning problems \cite{tandon2016gradient}. As a an example of this computing model, consider the data set $\mathbb{D}=\{X_i\}_{i=1}^{|\mathbb{D}|}$, where $X_i$ is the $i$th data sample in the set, and suppose we are interested in the sum of the respective evaluations of some function $f(.)$ over the points in the data set $\mathbb{D}$, i.e.,
\begin{equation*}
f(\mathbb{D})=\sum_{X_i\in\mathbb{D}}f(X_i).  
\end{equation*}
This summation could be split into $N$ sums (one for each worker):
\begin{equation}
f(\mathbb{D})=\sum_{X_i\in\CMcal{D}_1}f(X_i)+\sum_{X_i\in\CMcal{D}_2}f(X_i)+\dots+\sum_{X_i\in\CMcal{D}_N}f(X_i),
\label{computing model}
\end{equation}
where $\CMcal{D}_i$, $i=1,2,\dots,N$, is the $i$th partition of the data set $\mathbb{D}$. Now, if we assign the partition $\CMcal{D}_i$ to worker $i$ in System\ref{fig:sysModel1} and let each worker evaluate the function $f(.)$ on the data in its local subset followed by a summation over all evaluations, then the master node will be able to generate $f(\mathbb{D})$ by simply adding the local computations' results.

We model the redundant distribution of data among workers with a two-stage process. In the first stage, data samples get (redundantly) stored in equal-sized \textit{batches} and in the second stage batches get (redundantly) assigned to the workers. Batches overlap if each data sample is available in more than one batch. We assume the same batch size for the non-overlapping and the overlapping case, equal to $N/B$, where $B|N$, the number of batches is an integer in the range of $[B,N]$.

We define $T_{ij}$ to be the service time of worker $j$ on batch $i$. We assume that $T_{ij}$s are i.i.d random variables. We define the completion time of System\ref{fig:sysModel1} as the time window between the start of the computations at workers and the overall result generation at the master. We study the completion time under two class of distributions for $T_{ij}$, defined as follows.
\subsubsection{Exponential Distribution}
If $T_{ij} \sim \textup{Exp}(\mu)$ then,
    \begin{equation}
        \textup{Pr}\{T_{ij}>t\}=\mathbbm{1}(t\geq 0)\textup{e}^{-\mu t}
    \end{equation}
where $\mu$ is the service rate and is identical across the workers, and $\mathbbm{1}(.)$ is the indicator function.
\subsubsection{Shifted-Exponential Distribution}
If $T_{ij} \sim \textup{SExp}(\mu,\Delta)$ then,
    \begin{equation}
        \textup{Pr}\{T_{ij}>t\}=1-\mathbbm{1}(t\geq \Delta)\left[1-\textup{e}^{-\mu(t-\Delta)}\right]
    \end{equation}
where $\mu$ is the service rate and $\Delta$ is the shift parameter, defined as the minimum possible service time. We assume $\mu$ and $\Delta$ are identical across all the workers.

\section{Summary of Results}
In this section, we will provide the results of our analysis. We skip the proof of the theorems because of space limitations.
\begin{theorem}
With Exponential service time distribution, $T_{ij}\sim Exp(\mu)$, among all batch assignment policies, the balanced assignment of non-overlapping batches achieves the minimum expected job completion time.
\label{mainThm}
\end{theorem}

\begin{proof}
See \cite{behrouzi2018effect}.
\end{proof}

With a given batch size, the number of non-overlapping batches is fixed. Given that $B|N$, it is possible to assign every batch of data to an equal number of workers, which minimizes the average computing time. Moreover, for the given system model, with overlapping batches System\ref{fig:sysModel1} always has higher average completion time compared to balanced assignment of non-overlapping batches. Interested reader may find further discussions in \cite{behrouzi2018effect}.

\begin{corollary}
With Shifted-Exponential sevice time, $T_{ij}\sim \textup{SExp}(\Delta,\mu)$, the minimum expected job completion time is achieved by the balanced assignment of non-overlapping batches. 
\end{corollary}

It is proved in the original work \cite{behrouzi2018effect} that, the balanced assignment of non-overlapping batches minimized the expected completion time if the workers' service time is stochastically decreasing and convex random variable. Both Exponential and Shifted-Exponential random variables are stochastically decreasing (in the sense of usual stochastic ordering) and convex random variables. Therefore, for both distributions, assigning each data batch to $N/B$ workers minimizes the expected completion time.

We next study the effect of the level of redundancy on the completion time in System\ref{fig:sysModel1}, for the optimal data distribution for a given batch size. The level of redundancy can vary from full diversity, by assigning the entire data set to every worker, to full parallelism, by dividing the data set evenly among workers. We call this range the \textit{diversity-parallelism spectrum} and try to find the optimum operating point within this spectrum. To this end, we model the service time of each data sample $\tau$ with Exponential and Shifted-Exponential distributions. Further, to find the service time distribution of each batch, we use the size-dependent service time model proposed in \cite{gardner2016better}. In addition to the average, we also study the effect of the redundancy level on the variance of completion time. It is worth mentioning that the variance is considered to be even more important practical systems than the average job completion time, since it enables performance guarantees \cite{dean2013tail}.

\begin{theorem}
With Exponential service time distribution, $\tau\sim Exp(\mu)$, both the expected value and the variance of completion time is minimized in full diversity, i.e. $B=1$.
\end{theorem}

In the case of Exponential service time distribution of data samples, assignment of the entire data set to every worker could minimize the expected value and the variance of completion time. However, this may not be feasible in practice, for several reasons including, limited storage at workers or unexpected slowdowns due to overloading the entire system. In this case, our analysis suggest that the higher diversity is always beneficial for performance. Hence, with Exponential service time distribution of data samples, the highest possible diversity is optimal.

\begin{theorem}
With Shifted-Exponential service time distribution, $\tau\sim SExp(\Delta,\mu)$, the optimum operating point in the diversity-parallelism spectrum, from the expected completion time point of view, is the solution of the following discrete unconstrained optimization problem,
    \begin{equation}
        \underset{B\in F_B}{min}\qquad \frac{N\Delta}{B}+\frac{1}{\mu}H_B,
        \label{sexpExp}
    \end{equation}
where $F_B$ is the set of all feasible values for $B$. 
\end{theorem}

Form (\ref{sexpExp}), we see that there is an optimum value of $B$ which, for some set of parameters, is not at the either end of the diversity-parallelism spectrum. Similar dependency between the distribution's parameters and the efficient data placement was observed in coded storage systems \cite{joshi2012coding}. In Shifted-Exponential distribution, larger $\Delta\mu$ product means ``less randomness'' in the underlying random variable. With less randomness, it is more beneficial to spent resources to provide parallelism. Therefore, with larger $\Delta\mu$, higher parallelism is optimal. On the other hand, smaller $\Delta\mu$ means higher randomness and spending resources to provide diversity benefits performance more. This behavior of the completion time is showed in Fig. \ref{fig:sexpExpected}.

\begin{figure}[htbp]
   \centering
   \includegraphics[width=.9\columnwidth, keepaspectratio]{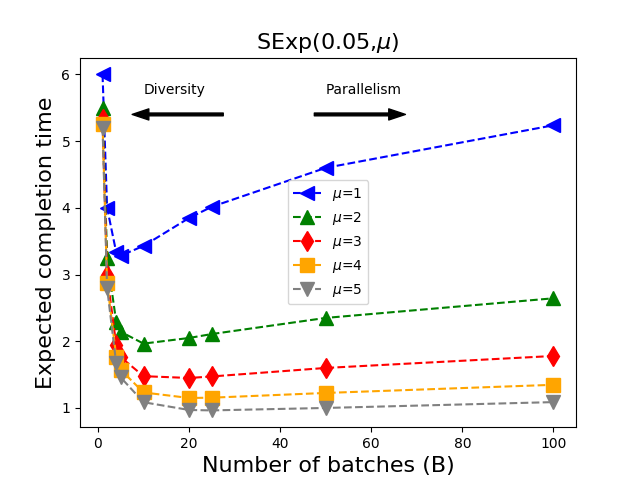}
   \caption{Expected completion time with Shifted-Exponential service time distribution of the data samples versus the number of batches, for different values of $\lambda$. With larger $\Delta\mu$ product higher diversity is optimal.}
   \label{fig:sexpExpected}
\end{figure}

\begin{theorem}
With Shifted Exponential service time distribution, $\tau\sim SExp(\Delta,\mu)$, the optimum operating point in the  diversity-parallelism spectrum, from the variance of completion time point of view, is minimized in full diversity.
\label{thm:sexpVarPred}
\end{theorem}

Achieving minimum variance at full diversity is in line with intuition, since the completion time in System\ref{fig:sysModel1} is a function of $B$ random variables, each being the minimum of $N/B$ workers' service time. By increasing diversity, i.e., decreasing B, the number of random variables from which we select a minimum increases, which in turn reduces the variability of the minimum value.

\section{conclusion}
In a distributed computing system with stragglers, optimum data replication policy was studied. The effect of the level of redundancy, ranging from full diversity to full parallelism, was studied for Exponential and Shifted-Exponential service time distribution of data samples. While with Exponential distribution full diversity is optimal, with Shifted-Exponential distribution the optimal operating point in diversity-parallelism spectrum depends of the distribution's parameters.

\section*{Acknowledgement}
Part of this research is based upon work supported by the NSF grants No.\ CIF-1717314 and CCF-1559855.

\bibliographystyle{IEEEtran}
\bibliography{ref}
\end{document}